\documentclass{article}

\usepackage{times}
\usepackage{amssymb,amsmath,amsthm}
\usepackage{mathrsfs}
\usepackage{epsfig}
\usepackage{color}

\newcommand\ie{{\em i.e.}~}
\newcommand\eg{{\em e.g.}~}

\def\B{\mathscr B}
\def\C{\mathbb C}
\def\D{\mathscr D}
\def\E{\mathcal E}
\def\F{\mathscr F}

\def\H{\mathcal H}

\def\M{\mathcal M}
\def\N{\mathbb N}

\def\R{\mathbb R}
\def\S{\mathscr S}

\def\V{\mathsf V}

\def\12{{\textstyle\frac12}}
\def\<{\left\langle}
\def\>{\right\rangle}
\def\({\left(}
\def\){\right)}
\def\[{\left[}
\def\]{\right]}
\def\dom{\mathcal D}
\def\lone{\mathsf{L}^{\:\!\!1}}

\def\ltwo{\mathsf{L}^{\:\!\!2}}

\def\one{\mathop{1\mskip-4mu{\rm l}}\nolimits}
\def\e{\mathop{\mathrm{e}}\nolimits}
\def\d{\mathrm{d}}

\def\sgn{\mathop{\mathrm{sgn}}\nolimits}

\newtheorem{Theorem}{Theorem}[section]
\newtheorem{Remark}[Theorem]{Remark}
\newtheorem{Lemma}[Theorem]{Lemma}
\newtheorem{Assumption}[Theorem]{Assumption}

\newtheorem{Definition}[Theorem]{Definition}
\newtheorem{Example}[Theorem]{Example}

\begin{document}


\title{{\Large\textbf{Anisotropic Lavine's formula and symmetrised time delay
\\in scattering theory}}}

\author{Rafael Tiedra de Aldecoa}
\date{\small
\begin{quote}
\emph{
\begin{itemize}
\item[] CNRS (UMR 8088) and Department of Mathematics, University of Cergy-Pontoise,
2 avenue Adolphe Chauvin, 95302 Cergy-Pontoise Cedex, France
\item[] \emph{E-mail:} rafael.tiedra@u-cergy.fr
\end{itemize}
}
\end{quote}
}

\maketitle


\begin{abstract}
We consider, in quantum scattering theory, symmetrised time delay defined in terms of sojourn times
in arbitrary spatial regions symmetric with respect to the origin. For potentials decaying more
rapidly than $|x|^{-4}$ at infinity, we show the existence of symmetrised time delay, and prove
that it satisfies an anisotropic version of Lavine's formula. The importance of an anisotropic
dilations-type operator is revealed in our study.
\end{abstract}

\section{Introduction and main results}\label{Intro}

It is known for long that the definition of {\em time delay} (in terms of sojourn times) in
scattering theory has to be {\em symmetrised} in the case of multichannel-type scattering processes
(see \eg \cite{AJ06,Bolle/Osborn,Martin75,Martin81,Smith60,Tiedra06}). More recently \cite{GT07} it
has been shown that symmetrised time delay does exist, in two-body scattering processes, for
arbitrary dilated spatial regions symmetric with respect to the origin (usual time delay does exist
only for spherical spatial regions \cite{Sassoli/Martin}). This leads to a generalised formula for
time delay, which reduces to the usual one in the case of spherical spatial regions. The aim of the
present paper is to provide a reasonable interpretation of this formula for potential scattering by
proving its identity with an anisotropic version of Lavine's formula \cite{Lav74}.

Let us recall the definition of symmetrised time delay for a two-body scattering process in $\R^d$,
$d\ge1$. Consider a bounded open set $\Sigma$ in $\R^d$ containing the origin and the dilated
spatial regions $\Sigma_r:=\{rx\mid x\in\Sigma\}$, $r>0$. Let $H_0:=-\12\Delta$ be the kinetic
energy operator in $\H:=\ltwo(\R^d)$ (endowed with the norm $\|\cdot\|$ and scalar product
$\<\cdot,\cdot\>$). Let $H$ be a selfadjoint perturbation of $H_0$ such that the wave operators $W_\pm:=\textrm{s-}\lim_{t\to\pm\infty}\e^{itH}\e^{-itH_0}$ exist and are complete (so that the
scattering operator $S:=W_+^*W_-$ is unitary). Then one defines for some states $\varphi\in\H$ and
$r>0$ two sojourn times, namely:
\begin{equation*}
T^0_r(\varphi)
:=\int_{-\infty}^\infty\d t\int_{x\in\Sigma_r}\d^dx\left|(\e^{-itH_0}\varphi)(x)\right|^2
\end{equation*}
and
\begin{equation*}
T_r(\varphi)
:=\int_{-\infty}^\infty\d t\int_{x\in\Sigma_r}\d^dx\left|(\e^{-itH}W_-\varphi)(x)\right|^2.
\end{equation*}
If the state $\varphi$ is normalized to one the first number is interpreted as the time spent by
the freely evolving state $\e^{-itH_0}\varphi$ inside the set $\Sigma_r$, whereas the second one
is interpreted as the time spent by the associated scattering state $\e^{-itH}W_-\varphi$ within
the same region. The usual time delay of the scattering process with incoming state $\varphi$ for
$\Sigma_r$ is defined as
$$
\tau_r^{\rm in}(\varphi):=T_r(\varphi)-T^0_r(\varphi),
$$
and the corresponding symmetrised time delay for $\Sigma_r$ is given by
$$
\tau_r(\varphi):=T_r(\varphi)-\12\[T^0_r(\varphi)+T^0_r(S\varphi)\].
$$
If $\Sigma$ is spherical and some abstract assumptions are verified, the limits of
$\tau_r^{\rm in}(\varphi)$ and $\tau_r(\varphi)$ as $r\to\infty$ exist and satisfy
\cite[Sec. 4.3]{GT07}
\begin{equation}\label{sweetie}
\lim_{r\to\infty}\tau_r(\varphi)
=\lim_{r\to\infty}\tau_r^{\rm in}(\varphi)
=-\12\big\langle H_0^{-1/2}\varphi,S^*[D,S]H_0^{-1/2}\varphi\big\rangle,
\end{equation}
where $D$ is the generator of dilations. If $\Sigma$ is not spherical the limit of
$\tau_r^{\rm in}(\varphi)$ as $r\to\infty$ does not exist anymore \cite{Sassoli/Martin}, but the
limit of $\tau_r(\varphi)$ as $r\to\infty$ does still exist, as soon as $\Sigma$ is symmetric with
respect to the origin \cite[Rem. 4.8]{GT07}.

In this paper we study $\tau_r(\varphi)$ in the setting of potential scattering. For potentials
decaying more rapidly than $|x|^{-4}$ at infinity, we prove the existence of
$\lim_{r\to\infty}\tau_r(\varphi)$ by using the results of \cite{GT07}. In a first step we show
that the limit satisfies the equality
\begin{equation}\label{tuptup}
\lim_{r\to\infty}\tau_r(\varphi)
=-\big\langle f(H_0)^{-1/2}\varphi,S^*[D_\Sigma,S]f(H_0)^{-1/2}\varphi\big\rangle,
\end{equation}
where $f$ is a real symbol of degree $1$ and $D_\Sigma\equiv D_\Sigma(f)$ is an operator acting as
an anisotropic generator of dilations. Then we prove that Formula \eqref{tuptup} can be rewritten
as an anisotropic Lavine's formula. Namely, one has (see Theorem \ref{lavine} for a precise
statement)
\begin{equation}\label{pizza_1}
\lim_{r\to\infty}\tau_r(\varphi)
=\int_{-\infty}^\infty\d s\,\big\langle \e^{-isH}W_-f(H_0)^{-1/2}\varphi,
\V_{\Sigma,f}\e^{-isH}W_-f(H_0)^{-1/2}\varphi\big\rangle,
\end{equation}
where the operator
$$
\V_{\Sigma,f}=f(H)-f(H_0)-i[V,D_\Sigma]
$$
generalises the virial $\widetilde V:=2V-i[V,D]$. Formula \eqref{pizza_1} provides an interesting
relation between the potential $V$ and symmetrised time delay, which we discuss.

Let us give a description of this paper. In section \ref{dilations} we introduce the condition on
the set $\Sigma$ (see Assumption \ref{Sigma}) under which our results are proved. We also define
the anisotropic generator of dilations $D_\Sigma$ and establish some of its properties. Section
\ref{potential_scattering} is devoted to symmetrised time delay in potential scattering; the
existence of symmetrised time delay for potentials decaying more rapidly than $|x|^{-4}$ at
infinity is established in Theorem \ref{time_delay}. In Theorem \ref{lavine} of Section
\ref{section_lavine} we prove the anisotropic Lavine's formula \eqref{pizza_1} for the same class
of potentials. Remarks and examples are to be found at the end of Section \ref{section_lavine}.

We emphasize that the extension of Lavine's formula to non spherical sets $\Sigma$ is not
straightforward due, among other things, to the appearance of a singularity in the space of
momenta not present in the isotropic case (see Equation \eqref{partial_alpha} and the paragraphs
that follow). The adjunction of the symbol $f$ in the definition of the operator $D_\Sigma$ (see
Definition \ref{F_Sigma}) is made to circumvent the difficulty.

Finally we refer to \cite{Jen84} (see also \cite{Jen83,Lav74,Nak87,Nar80,Nar84}) for a related
work on Lavine's formula for time delay.

\section{Anisotropic dilations}\label{dilations}

In this section we define the operator $D_\Sigma$ and establish some of its properties in relation
with the generator of dilations $D$ and the shape of $\Sigma$. We start by recalling some
notations.

Given two Hilbert spaces $\H_1$ and $\H_2$, we write $\B(\H_1,\H_2)$ for the set of bounded
operators from $\H_1$ to $\H_2$ with norm $\|\cdot\|_{\H_1\to\H_2}$, and put
$\B(\H_1):=\B(\H_1,\H_1)$. We set $Q:=(Q_1,Q_2,\ldots,Q_d)$ and $P:=(P_1,P_2,\ldots,P_d)$, where
$Q_j$ (resp. $P_j$) stands for the $j$-th component of the position (resp. momentum) operator in
$\H$. $\N:=\{0,1,2,\ldots\}$ is the set of natural numbers. $\H^k$, $k\in\N$, are the usual Sobolev
spaces over $\R^d$, and $\H^s_t(\R^d)$, $s,t\in\R$, are the weighted Sobolev spaces over $\R^d$
\cite[Sec. 4.1]{ABG}, with the convention that $\H^s(\R^d):=\H^s_0(\R^d)$ and
$\H_t(\R^d):=\H^0_t(\R^d)$. Given a set $\M\subset\R^d$ we write $\one_\M$ for the characteristic
function for $\M$. We always assume that $\Sigma$ is a bounded open set in $\R^d$ containing $0$,
with boundary $\partial\Sigma$ of class $C^4$. Often we even suppose that $\Sigma$ satisfies the
following stronger assumption (see \cite[Sec. 2]{GT07}).

\begin{Assumption}\label{Sigma}
$\Sigma$ is a bounded open set in $\R^d$ containing $0$, with boundary $\partial\Sigma$ of class
$C^4$. Furthermore $\Sigma$ satifies
$$
\int_0^\infty\d\mu\[\one_\Sigma(\mu x)-\one_\Sigma(-\mu x)\]=0,\quad\forall x\in\R^d.
$$
\end{Assumption}
\noindent
If $p\in\R^d$, then the number $\int_0^\infty\d t\one_\Sigma(tp)$ is the sojourn time in $\Sigma$ of
a free classical particle moving along the trajectory $t\mapsto x(t):=tp$, $t\ge0$. Obviously
$\Sigma$ satisfies Assumption \ref{Sigma} if $\Sigma$ is symmetric with respect to $0$ (\ie
$\Sigma=-\Sigma$). Moreover if $\Sigma$ is star-shaped with respect to $0$ and satisfies
Assumption \ref{Sigma}, then $\Sigma=-\Sigma$.

We recall from \cite[Lemma 2.2]{GT07} that the limit
\begin{equation}\label{R_Sigma}
R_\Sigma(x):=\lim_{\varepsilon\searrow0}\(\int_\varepsilon^{+\infty}\frac{\d\mu}\mu
\one_\Sigma(\mu x)+\ln\varepsilon\)
\end{equation}
exists for each $x\in\R^d\setminus\{0\}$, and we define the function
$G_\Sigma:\R^d\setminus\{0\}\to\R$ by
\begin{equation}\label{G_Sigma}
G_\Sigma(x):=\12\[R_\Sigma(x)+R_\Sigma(-x)\].
\end{equation}
The function $G_\Sigma:\R^d\setminus\{0\}\to\R$ is of class $C^4$ since $\partial\Sigma$ of class
$C^4$. Let $x\in\R^d\setminus\{0\}$ and $t>0$, then Formulas \eqref{R_Sigma} and \eqref{G_Sigma}
imply that
$$
G_\Sigma(tx)=G_\Sigma(x)-\ln(t).
$$
From this one easily gets the following identities for the derivatives of $G_\Sigma$:
\begin{align}
x\cdot(\nabla G_\Sigma)(x)&=-1,\label{moinsun}\\
t^{|\alpha|}\big(\partial^\alpha G_\Sigma\big)(tx)&=\big(\partial^\alpha G_\Sigma\big)(x),
\label{partial_alpha}
\end{align}
where $\alpha$ is a $d$-dimensional multi-index with $|\alpha|\ge1$ and $\partial^\alpha:=\partial_1^{\alpha_1}\cdots\partial_d^{\alpha_d}$. The second identity
suggests a way of regularizing the functions $\partial_jG_\Sigma$ which partly motivates the
following definition. We use the notation $S^\mu(\R;\R)$, $\mu\in\R$, for the vector space of
real symbols of degree $\mu$ on $\R$.

\begin{Definition}\label{F_Sigma}
Let $f\in S^1(\R;\R)$ be such that
\begin{enumerate}
\item[(i)] $f(0)=0$ and $f(u)>0$ for each $u>0$,
\item[(ii)] for each $j=1,2,\ldots,d$, the function $x\mapsto(\partial_jG_\Sigma)(x)f(x^2/2)$
(a priori only defined for $x\in\R^d\setminus\{0\}$) belongs to $C^3(\R^d;\R)$.
\end{enumerate}
Then we define $F_\Sigma:\R^d\to\R^d$ by $F_\Sigma(x):=-(\nabla G_\Sigma)(x)f(x^2/2)$.
\end{Definition}

Given a set $\Sigma$ there are many appropriate choices for the function $f$. For instance if
$\gamma>0$ one can always take $f(u)=2(u^2+\gamma)^{-1}u^3$, $u\in\R$. But when
$\Sigma$ is equal to the open unit ball $\mathcal B:=\{x\in\R^d\mid|x|<1\}$ one can obviously
make a simpler choice. Indeed in such case one has \cite[Rem. 2.3.(b)]{GT07}
$(\partial_jG_{\mathcal B})(x)=-x_jx^{-2}$, and the choice $f(u)=2u$, $u\in\R$, leads to the
$C^\infty$-function $F_\Sigma(x)=x$.

\begin{Remark}\label{orthogonal}
One can associate to each set $\Sigma$ a unique set $\widetilde\Sigma$ symmetric and star-shaped
with respect to $0$ such that $G_\Sigma=G_{\widetilde\Sigma}$ \cite[Rem. 2.3.(c)]{GT07}. The
boundary $\partial\widetilde\Sigma$ of $\widetilde\Sigma$ satisfies
$$
\partial\widetilde\Sigma:=\big\{\e^{G_\Sigma(x)}x\mid x\in\R^d\setminus\{0\}\big\},
$$
and $\widetilde\Sigma_r:=\big\{rx\mid x\in\widetilde\Sigma\big\}$, $r>0$. Thus the vector field
$F_\Sigma=F_{\widetilde\Sigma}$ is orthogonal to the hypersurfaces $\partial\widetilde\Sigma_r$
in the following sense: if $v$ belongs to the tangent space of $\partial\widetilde\Sigma_r$ at
$y\in\partial\widetilde\Sigma_r$, then $F_\Sigma(y)$ is orthogonal to $v$. To see this let
$s\mapsto y(s)\equiv r\e^{G_\Sigma(x(s))}x(s)$ be any differentiable curve on
$\partial\widetilde\Sigma_r$. Then $\frac\d{\d s}\;\!y(s)$ belongs to the tangent space of
$\partial\widetilde\Sigma_r$ at $y(s)$, and a direct calculation using Equations
\eqref{moinsun}-\eqref{partial_alpha} gives $F_\Sigma(y(s))\cdot\frac\d{\d s}\;\!y(s)=0$.
\end{Remark}

In the rest of the section we give a meaning to the expression
$$
D_\Sigma:=\12[F_\Sigma(P)\cdot Q+Q\cdot F_\Sigma(P)],
$$
and we establish some properties of $D_\Sigma$ in relation with the generator of dilations
$$
D:=\12(P\cdot Q+Q\cdot P).
$$
For the next lemma we emphasize that $\H^2$ is contained in the domain $\dom\big(f(H_0)\big)$
of $f(H_0)$. The notation $\<\:\!\cdot\:\!\>$ stands for $\sqrt{1+|\cdot|^2}$, and $\S$ is
the Schwartz space on $\R^d$.

\begin{Lemma}\label{D_Sigma}
Let $\Sigma$ be a bounded open set in $\R^d$ containing $0$, with boundary $\partial\Sigma$ of
class $C^4$. Then
\begin{enumerate}
\item[(a)] The operator $D_\Sigma$ is essentially selfadjoint on $\S$. As a bounded operator,
$D_\Sigma$ extends to an element of $\B\big(\H^s_t,\H^{s-1}_{t-1}\big)$ for each $s\in\R$,
$t\in[-2,0]\cup[1,3]$.
\item[(b)] One has for each $t\in\R$ and $\varphi\in\dom(D_\Sigma)\cap\dom\big(f(H_0)\big)$
\begin{equation}\label{group_com}
\e^{-itH_0}D_\Sigma\e^{itH_0}\varphi=[D_\Sigma-tf(H_0)]\varphi.
\end{equation}
In particular one has the equality
\begin{equation}\label{gen_com}
i[H_0,D_\Sigma]=f(H_0)
\end{equation}
as sesquilinear forms on $\dom(D_\Sigma)\cap\H^2$.
\end{enumerate}
\end{Lemma}

The second claim of point (a) is sufficient for our purposes, even if it is only a particular
case of a more general result.

\begin{proof}
(a) The essential seladjointness of $D_\Sigma$ on $\S$ follows from the fact that $F_\Sigma$
is of class $C^3$ (see \eg \cite[Prop. 7.6.3.(a)]{ABG}).

Due to the hypotheses on $F_\Sigma$ one has for each $\varphi\in\S$ the bound
\begin{equation}\label{joliebound}
\big\|(\partial^\alpha{F_\Sigma}_j)(P)\varphi\|\le{\rm Const.}\left\|\<P\>\varphi\right\|,
\end{equation}
where ${F_\Sigma}_j$ is the $j$-th component of $F_\Sigma$ and $\alpha$ is a $d$-dimensional
multi-index with $|\alpha|\le3$. Furthermore
$$
\|D_\Sigma\|_{\H^s_3\to\H^{s-1}_2}
\le\sum_{j\le d}\sup_{\varphi\in\S,\|\varphi\|_{\H^s_3}=1}
\big\|\<P\>^{s-1}\<Q\>^2\[{F_\Sigma}_j(P)Q_j+\hbox{$\frac i2$}(\partial_j{F_\Sigma}_j)(P)\]
\varphi\big\|
$$
for each $s\in\R$. Since $\<Q\>^2$ acts as the operator $1-\Delta$ after a Fourier transform,
the inequalities above imply that $D_\Sigma$ extends to an element of $\B(\H^s_3,\H^{s-1}_2)$.
A similar argument shows that $D_\Sigma$ extends to an element of $\B(\H^s_1,\H^{s-1})$ for
each $s\in\R$. The second part of the claim follows then by using interpolation and duality.

(b) Let $\varphi\in\e^{-itH_0}\S$. Since $\e^{-itH_0}Q_j\e^{itH_0}\varphi=(Q_j-tP_j)\varphi$, it
follows by Formula \eqref{moinsun} that
$$
\e^{-itH_0}D_\Sigma\e^{itH_0}\varphi=[D_\Sigma+tP\cdot(\nabla G_\Sigma)(P)f(H_0)]\varphi
=[D_\Sigma-tf(H_0)]\varphi.
$$
This together with the essential selfajointness of $\e^{-itH_0}D_\Sigma\e^{itH_0}$ on
$\e^{-itH_0}\S$ implies the first part of the claim. Relation \eqref{gen_com} follows by
taking the derivative of \eqref{group_com} w.r.t. $t$ in the form sense and then posing $t=0$.
\end{proof}

\begin{Remark}\label{about_the_group}
If $\Sigma=\mathcal B$ and $f(u)=2u$, then $F_\Sigma(x)=x$ for each $x\in\R^d$, and the
operators $D_\Sigma$ and $D$ coincide. If $\Sigma$ is not spherical it is still possible to
determine part of the behaviour of the group $W_t:=\e^{itD_\Sigma}$. Indeed let
$\R\times\R^d\ni(t,x)\mapsto\xi_t(x)\in\R^d$
be the flow associated to the vector field $-F_\Sigma$, that is, the solution of the
differential equation
\begin{equation}\label{differential}
\frac\d{\d t}\,\xi_t(x)=(\nabla G_\Sigma)(\xi_t(x))f\big(\xi_t(x)^2/2\big),\quad\xi_0(x)=x.
\end{equation}
Then it is known (see \eg the proof of \cite[Prop. 7.6.3.(a)]{ABG}) that the group $W_t$ acts
in the Fourier space as
\begin{equation}\label{InFourier}
\big(\widehat W_t\varphi\big)(x):=\sqrt{\eta_t(x)}\varphi(\xi_t(x)),
\end{equation}
where $\eta_t(x)\equiv\det(\nabla\xi_t(x))$ is the Jacobian at $x$ of the mapping
$x\mapsto\xi_t(x)$. Taking the scalar product of Equation \eqref{differential} with $\xi_t(x)$
and then using Formula \eqref{moinsun} leads to the equation
$$
\frac\d{\d t}\,\xi_t(x)^2=-2f\big(\xi_t(x)^2/2\big),\quad\xi_0(x)=x.
$$
If $t<0$ and $x\ne0$, then $\xi_t(x)^2\ge x^2>0$, and $\xi_t(x)^2$ is given by the implicit
formula
$$
2t+\int_{x^2}^{\xi_t(x)^2}\d u\,f(u/2)^{-1}=0.
$$
This, together with the facts that $x\mapsto f(x^2/2)$ belongs to $S^2(\R;\R)$ and $f(u)>0$ for
$u>0$, implies the estimate $\<\xi_t(x)\>\le\e^{-\textsc ct}\<x\>$ for some constant
$\textsc c>0$. Since $\<\xi_t(x)\>\le\<x\>$ for each $t\ge0$ it follows that
\begin{equation}\label{upperbound}
\<\xi_t(x)\>\le(1+\e^{-\textsc ct})\<x\>
\end{equation}
for all $t\in\R$ and $x\in\R^d$ (the case $x=0$ is covered since $\xi_t(0)=0$ for all $t\in\R$).
Equation \eqref{upperbound} implies that the domain $\H^2$ of $H_0$ is left invariant by the
group $W_t$.
\end{Remark}

The results of Remarks \ref{orthogonal} and \ref{about_the_group} suggest that $W_t$ may be
interpreted as an anisotropic version of the dilations group, which reduces to the usual
dilations group in the case $\Sigma=\mathcal B$ and $f(u)=2u$.

In the next lemma we show some properties of the mollified resolvent
$$
R_\lambda:=i\lambda(D_\Sigma+i\lambda)^{-1},\quad\lambda\in\R\setminus\{0\}.
$$
We refer to \cite[Lemma 6.2]{PSS81} for the same results in the case of the usual dilations
generator $D$, that is, when $\Sigma=\mathcal B$ and $f(u)=2u$. See also \cite[Lemma 4.5]{CFKS}
for a general result.

\begin{Lemma}\label{R_lambda}
Let $\Sigma$ be a bounded open set in $\R^d$ containing $0$, with boundary $\partial\Sigma$
of class $C^4$. Then
\begin{enumerate}
\item[(a)] One has for each $t\in\R$ and $\varphi\in\dom\big(\xi_t(P)^2\big)$
\begin{equation}\label{D-homogeneous}
\e^{itD_\Sigma}H_0\e^{-itD_\Sigma}\varphi=\12\;\!\xi_t(P)^2\varphi.
\end{equation}

\item[(b)] For each $\lambda\in\R$ with $|\lambda|$ large enough, $R_\lambda$ belongs to
$\B(\H^2)$, and $R_\lambda$ extends to an element of $\B(\H^{-2})$. Furthermore we have for
each $\varphi\in\H^2$ and each $\psi\in\H^{-2}$
$$
\lim_{|\lambda|\to\infty}\|(1-R_\lambda)\varphi\|_{\H^2}=0\quad{\rm and}\quad
\lim_{|\lambda|\to\infty}\|(1-R_\lambda)\psi\|_{\H^{-2}}=0.
$$
\end{enumerate}
\end{Lemma}

\begin{proof}
(a) Let $\varphi\in \e^{itD_\Sigma}\S$. A direct calculation using Formula \eqref{InFourier}
gives
$$
(\F\e^{itD_\Sigma}H_0\e^{-itD_\Sigma}\varphi)(k)=\12\;\!\xi_t(k)^2(\F\varphi)(k),
$$
where $\F$ is the Fourier transform. This together with the essential selfajointness of
$\e^{itD_\Sigma}H_0\e^{-itD_\Sigma}$ on $\e^{itD_\Sigma}\S$ implies the claim.

(b) Let $\varphi\in\H^2$ and take $\lambda\in\R$ with $|\lambda|>\textsc c$, where $\textsc c$
is the constant in the inequality \eqref{upperbound}. Using the (strong) integral formula
$$
(D_\Sigma+i\lambda)^{-1}=i\int_0^{\mp\infty}\d t\,\e^{\lambda t}\e^{-itD_\Sigma},\quad
\sgn(\lambda)=\pm1,
$$
and Relation \eqref{D-homogeneous} we get the equalities
\begin{align*}
(D_\Sigma+i\lambda)^{-1}\varphi&=(H_0+1)^{-1}(D_\Sigma+i\lambda)^{-1}(H_0+1)\varphi\\
&\qquad+i\int_0^{\mp\infty}\d t\,\e^{\lambda t}\[\e^{-itD_\Sigma},(H_0+1)^{-1}\](H_0+1)\varphi\\
&=(H_0+1)^{-1}(D_\Sigma+i\lambda)^{-1}(H_0+1)\varphi\\
&\qquad-i\int_0^{\mp\infty}\d t\,
\e^{\lambda t}(H_0+1)^{-1}\e^{-itD_\Sigma}\[H_0-\12\xi_t(P)^2\]\varphi\\
&=(H_0+1)^{-1}(D_\Sigma+i\lambda)^{-1}\varphi
+\hbox{$\frac i2$}(H_0+1)^{-1}\int_0^{\mp\infty}\d t\,
\e^{\lambda t}\e^{-itD_\Sigma}\xi_t(P)^2\varphi.
\end{align*}
It follows that
$$
H_0R_\lambda\varphi=-\hbox{$\frac\lambda2$}
\int_0^{\mp\infty}\d t\,\e^{\lambda t}\e^{-itD_\Sigma}\xi_t(P)^2\varphi,\quad\sgn(\lambda)=\pm1.
$$
Now $|\lambda|>\textsc c$, and
$\left\|\xi_t(P)^2\varphi\right\|\le(1+\e^{-\textsc ct})\|\varphi\|_{\H^2}$ due to the bound
\eqref{upperbound}. Thus
\begin{align}
\|H_0R_\lambda\varphi\|&\le\hbox{$\frac{|\lambda|}2$}
\int_0^\infty\d t\,\e^{-|\lambda|t}\left\|\xi_{-\sgn(\lambda)t}(P)^2\varphi\right\|\nonumber\\
&\le\hbox{$\frac{|\lambda|}2$}
\int_0^\infty\d t\,\big(\e^{-|\lambda|t}+\e^{(\sgn(\lambda)\textsc c-|\lambda|)t}\big)
\|\varphi\|_{\H^2}\nonumber\\
&\le{\rm Const.}\;\!\|\varphi\|_{\H^2}.\label{H0Rlambda}
\end{align}
Using the estimate \eqref{H0Rlambda} and a duality argument one gets the bounds
\begin{equation}\label{bound+2-2}
\|R_\lambda\|_{\H^2\to\H^2}\le{\rm Const.}\qquad{\rm and}\qquad
\|R_\lambda\|_{\H^{-2}\to\H^{-2}}\le{\rm Const.},
\end{equation}
which imply the first part of the claim. For the second part we remark that
$$
1-R_\lambda=(i\lambda)^{-1}D_\Sigma R_\lambda
$$
on $\H$. Using this together with the bounds \eqref{bound+2-2} one easily shows that $\lim_{|\lambda|\to\infty}\|(1-R_\lambda)\varphi\|_{\H^2}=0$ for each $\varphi\in\H^2$ and that $\lim_{|\lambda|\to\infty}\|(1-R_\lambda)\psi\|_{\H^{-2}}=0$ for each $\psi\in\H^{-2}$.
\end{proof}

\section{Symmetrised time delay}\label{potential_scattering}

In this section we collect some facts on short-range scattering theory in connection with the
existence of symmetrised time delay. We always assume that the potential $V$ satisfies the usual
Agmon-type condition:

\begin{Assumption}\label{potential}
$V$ is a multiplication operator by a real-valued function such that $V$ defines a compact operator
from $\H^2$ to $\H_\kappa$ for some $\kappa>1$.
\end{Assumption}

By using duality, interpolation and the fact that $V$ commutes with the operator $\<Q\>^t$,
$t\in\R$, one shows that $V$ also defines a bounded operator from $\H^{2s}_t$ to
$\H^{2(s-1)}_{t+\kappa}$ for any $s\in[0,1]$, $t\in\R$. Furthermore the operator sum $H:=H_0+V$ is
selfadjoint on $\dom(H)=\H^2$, the wave operators $W_\pm$ exist and are complete, and the
projections $\one_{\Sigma_r}(Q)$ are locally $H$-smooth on $(0,\infty)\setminus\sigma_{\rm pp}(H)$
(see \eg \cite[Sec. 3]{Jensen81} and \cite[Sec. XIII.8]{RSIV}).

Since the first two lemmas are somehow standard, we give their proofs in the appendix.

\begin{Lemma}\label{Hminus}
Let $V$ satisfy Assumption \ref{potential} with $\kappa>1$, and take
$z\in\C\setminus\{\sigma(H_0)\cup\sigma(H)\}$. Then the operator $(H-z)^{-1}$ extends to an element
of $\B\big(\H^{-2s}_t,\H^{2(1-s)}_t\big)$ for each $s\in[0,1]$, $t\in\R$.
\end{Lemma}

Alternate formulations of the next lemma can be found in \cite[Lemma 4.6]{Jensen81} and
\cite[Lemma 3.9]{Tiedra06}. For each $s\ge0$ we define the dense set
$$
\D_s:=\big\{\varphi\in\dom(\<Q\>^s)\mid\eta(H_0)\varphi=\varphi\textrm{ for some }\eta\in
C^\infty_0((0,\infty)\setminus\sigma_{\rm pp}(H))\big\}.
$$

\begin{Lemma}\label{cond_lone}
Let $V$ satisfy Assumption \ref{potential} with $\kappa>2$. Then one has for each $\varphi\in\D_s$
with $s>2$
\begin{equation}
\left\|(W_--1)\e^{-itH_0}\varphi\right\|\in\lone(\R_-,\d t)\label{R-}
\end{equation}
and
\begin{equation}
\left\|(W_+-1)\e^{-itH_0}\varphi\right\|\in\lone(\R_+,\d t).\label{R+}
\end{equation}
\end{Lemma}

\begin{Lemma}\label{S_mapping}
Let $V$ satisfy Assumption \ref{potential} with $\kappa>4$, and let $\varphi\in\D_s$ for some
$s>2$. Then there exists $s'>2$ such that $S\varphi\in\D_{s'}$, and the following conditions
are satisfied:
$$
\left\|(W_--1)\e^{-itH_0}\varphi\right\|\in\lone(\R_-,\d t)\quad{\rm and}\quad
\left\|(W_+-1)\e^{-itH_0}S\varphi\right\|\in\lone(\R_+,\d t).
$$
\end{Lemma}

\begin{proof}
The first part of the claim follows by \cite[Thm. 1.4.(ii)]{Jensen/Nakamura}. Since
$\varphi\in\D_s$ and $S\varphi\in\D_{s'}$ with $s,s'>2$, the second part of the claim follows
by Lemma \ref{cond_lone}.
\end{proof}

\begin{Theorem}\label{time_delay}
Let $\Sigma$ satisfy Assumption \ref{Sigma}. Suppose that $V$ satisfies Assumption
\ref{potential} with $\kappa>4$. Let $\varphi\in\D_s$ with $s>2$. Then the limit of
$\tau_r(\varphi)$ as $r\to\infty$ exists, and one has
\begin{equation}\label{eq_time}
\lim_{r\to\infty}\tau_r(\varphi)
=-\big\langle f(H_0)^{-1/2}\varphi,S^*[D_\Sigma,S]f(H_0)^{-1/2}\varphi\big\rangle.
\end{equation}
\end{Theorem}

\begin{proof}
Due to Lemma \ref{S_mapping} all the assumptions for the existence of
$\lim_{r\to\infty}\tau_r(\varphi)$ are verified (see \cite[Sec. 4]{GT07}), and we know by
Theorem \cite[Thm. 4.6]{GT07} that
$$
\lim_{r\to\infty}\tau_r(\varphi)=-\12\<\varphi,S^*\[i[Q^2,G_\Sigma(P)],S\]\varphi\>.
$$
It follows that
\begin{align*}
\lim_{r\to\infty}\tau_r(\varphi)
&=\12\<\varphi,S^*[Q\cdot(\nabla G_\Sigma)(P)+(\nabla G_\Sigma)(P)\cdot Q,S]\varphi\>\\
&=\12\big\langle f(H_0)^{-1/2}\varphi,S^*\big[f(H_0)^{1/2}\big(Q\cdot(\nabla G_\Sigma)(P)\\
&\qquad\qquad\qquad\qquad
+(\nabla G_\Sigma)(P)\cdot Q\big)f(H_0)^{1/2},S\big]f(H_0)^{-1/2}\varphi\big\rangle\\
&=-\big\langle f(H_0)^{-1/2}\varphi,S^*[D_\Sigma,S]f(H_0)^{-1/2}\varphi\big\rangle.
\end{align*}
\end{proof}

Note that Theorem \ref{time_delay} can be proved with the function $f(u)=2u$, even if $\Sigma$ is
not spherical. Indeed, in such a case, point (ii) of Definition \ref{F_Sigma} is the only
assumption not satisfied by $f$, and a direct inspection shows that this assumption does not play
any role in the proof of Theorem \ref{time_delay}.

\begin{Remark}
Some results of the literature suggest that Theorem \ref{time_delay} may be proved under a less
restrictive decay assumption on $V$ if one modifies some of the previous definitions. Typically
one proves the existence of (usual) time delay for potentials decaying more rapidly than
$|x|^{-2}$ (or even $|x|^{-1}$) at infinity by using smooth cutoff in configuration space and by
considering particular potentials. The reader is referred to \cite{ACS,MSA92,Nak87,Wan87,Wan88}
for more informations on this issue.
\end{Remark}

\section{Anisotropic Lavine's formula}\label{section_lavine}

In this section we prove the anisotropic Lavine's formula \eqref{pizza_1}. We first give a precise
meaning to some commutators.

\begin{Lemma}\label{D_SigmaV}
Let $\Sigma$ be a bounded open set in $\R^d$ containing $0$ with boundary $\partial\Sigma$ of class
$C^4$. Let $V$ satisfy Assumption \ref{potential} with $\kappa>1$. Then
\begin{enumerate}
\item[(a)] The commutator $[V,D_\Sigma]$, defined as a sesquilinear form on
$\dom(D_\Sigma)\cap\H^2$, extends uniquely to an element of $\B(\H^2,\H^{-2})$.
\item[(b)] For each $t\in\R$ the commutator $[D_\Sigma,\e^{-itH}]$, defined as a sesquilinear form
on $\dom(D_\Sigma)\cap\H^2$, extends uniquely to an element $[D_\Sigma,\e^{-itH}]^a$ of
$\B(\H^2,\H^{-2})$ which satisfies
$$
\big\|[D_\Sigma,\e^{-itH}]^a\big\|_{\H^2\to\H^{-2}}\le{\rm Const.}\;\!|t|.
$$
\item[(c)] For each $\eta\in C^\infty_0(\R)$ the commutator $[D_\Sigma,\eta(H)]$, defined as a
sesquilinear form on $\dom(D_\Sigma)\cap\H^2$, extends uniquely to an element of $\B(\H)$. In
particular, the operator $\eta(H)$ leaves $\dom(D_\Sigma)$ invariant.
\end{enumerate}
\end{Lemma}

\begin{proof}
Point (a) follows easily from Lemma \ref{D_Sigma}.(a) and the hypotheses on $V$. Given point (a)
and Lemma \ref{R_lambda}.(b), one shows points (b) and (c) as in \cite[Lemma 7.4]{PSS81}.
\end{proof}

If $V$ satisfies Assumption \ref{potential} with $\kappa>2$, then the result of Lemma
\ref{D_SigmaV}.(a) can be improved by using Lemma \ref{D_Sigma}.(a). Namely, there exists
$\delta>\12$ such that the commutator $[V,D_\Sigma]$, defined as a sesquilinear form on
$\dom(D_\Sigma)\cap\H^2$, extends uniquely to an element $[V,D_\Sigma]^a$ of
$\B\big(\H^2_{-\delta},\H^{-2}_\delta\big)$.

Next Lemma is a generalisation of \cite[Lemmas 2.5 \& 2.7]{Jen84}. It is proved under the following
assumption on the function $f$.

\begin{Assumption}\label{difference}
For each $t\in\R$
there exists $\rho>1$ such that the operator $f(H)-f(H_0)$, defined on $\H^2$,
extends to an element of $\B\big(\H^2_t,\H_{t+\rho}\big)$.
\end{Assumption}

We refer to Remark \ref{bibo} for examples of admissible functions $f$. Here we only note that the
operator
$$
\V_{\Sigma,f}:=f(H)-i[H,D_\Sigma]^a=f(H)-f(H_0)-i[V,D_\Sigma]^a.
$$
belongs to $\B(\H^2_{-\delta},\H^{-2}_\delta)$ for some $\delta>\12$ as soon as $f$ satisfies
Assumption \ref{difference}.

\begin{Lemma}\label{giantcoucou}
Let $\Sigma$ be a bounded open set in $\R^d$ containing $0$, with boundary $\partial\Sigma$ of
class $C^4$. Let $V$ satisfy Assumption \ref{potential} with $\kappa>2$. Suppose that Assumption
\ref{difference} is verified. Then
\begin{enumerate}
\item[(a)] One has for each $\eta\in C^\infty_0((0,\infty)\setminus\sigma_{\rm pp}(H))$ and each
$t\in\R$ the inequality
$$
\big\|(D_\Sigma+i)^{-1}\e^{-itH}\eta(H)(D_\Sigma+i)^{-1}\big\|\le{\rm Const.}\<t\>^{-1}.
$$
\item[(b)] For each $\eta\in C^\infty_0((0,\infty)\setminus\sigma_{\rm pp}(H))$ the operators
$[D_\Sigma,W_\pm\eta(H_0)]$ and $[D_\Sigma,W_\pm^*\eta(H)]$, defined as sesquilinear forms on
$\dom(D_\Sigma)$, extend uniquely to elements of $\B(\H)$. In particular, the operators
$W_\pm\eta(H_0)$ and $W_\pm^*\eta(H)$ leave $\dom(D_\Sigma)$ invariant.
\end{enumerate}
\end{Lemma}

\begin{proof} (a) Since the case $t=0$ is trivial, we can suppose $t\ne0$. Let
$\varphi,\psi\in\dom(D_\Sigma)\cap\H^2$, then
$$
\<D_\Sigma\varphi,\e^{-itH}\psi\>-\<\varphi,\e^{-itH}D_\Sigma\psi\>=\lim_{\lambda\to\infty}
\int_0^t\d s\,\big\langle\varphi,\e^{i(s-t)H}i[H,D_\Sigma R_\lambda]\e^{-isH}\psi\big\rangle
$$
due to Lemma \ref{R_lambda}.(b). By using Lemma \ref{D_Sigma}.(b) and Lemma \ref{D_SigmaV}.(b) we
get in $\B(\H^2,\H^{-2})$ the equalities
\begin{align}
[D_\Sigma,\e^{-itH}]^a&=\e^{-itH}\int_0^t\d s\,\e^{isH}i[H,D_\Sigma]^a\e^{-isH}\nonumber\\
&=t\e^{-itH}f(H)-\e^{-itH}\int_0^t\d s\,\e^{isH}\V_{\Sigma,f}\e^{-isH}.\label{erg}
\end{align}
Take $\eta,\vartheta\in C^\infty_0((0,\infty)\setminus\sigma_{\rm pp}(H))$ with $\vartheta$
identically one on the support of $\eta$, and let
$\zeta\in C^\infty_0((0,\infty)\setminus\sigma_{\rm pp}(H))$ be defined by
$\zeta(u):=f(u)^{-1}\vartheta(u)$. Then $\eta(H)=f(H)\zeta(H)\eta(H)$ and
\begin{align*}
\e^{-itH}\eta(H)&=\frac1t\zeta(H)t\e^{-itH}f(H)\eta(H)\\
&=\frac1t\zeta(H)\e^{-itH}\int_0^t\d s\,\e^{isH}\V_{\Sigma,f}\e^{-isH}\eta(H)
+\frac1t\zeta(H)[D_\Sigma,\e^{-itH}]^a\eta(H).
\end{align*}
Since $\V_{\Sigma,f}$ belongs to $\B(\H^2_{-\delta},\H^{-2}_\delta)$ for some $\delta>\12$, a
local $H$-smoothness argument shows that the first term is bounded by ${\rm Const.}|t|^{-1}$ in
$\H$. Furthermore by using Lemma \ref{D_SigmaV}.(c) one shows that $(D_\Sigma+i)^{-1}\zeta(H)[D_\Sigma,\e^{-itH}]^a\eta(H)(D_\Sigma+i)^{-1}$ is bounded in $\H$ by
a constant independent of $t$. Thus
$$
\left\|(D_\Sigma+i)^{-1}\e^{-itH}\eta(H)(D_\Sigma+i)^{-1}\right\|\le{\rm Const.}\;\!|t|^{-1},
$$
and the claim follows.

(b) Consider first $[D_\Sigma,W_+\eta(H_0)]$. Given
$\eta\in C^\infty_0((0,\infty)\setminus\sigma_{\rm pp}(H))$ let
$\zeta\in C^\infty_0((0,\infty)\setminus\sigma_{\rm pp}(H))$ be identically one on the support of
$\eta$. Due to Lemma \ref{D_SigmaV}.(c) one has on $\dom(D_\Sigma)$
\begin{align*}
&[D_\Sigma,\zeta(H)\e^{itH}\eta(H)\e^{-itH_0}\zeta(H_0)]\\
&=\zeta(H)[D_\Sigma,\e^{itH}\eta(H)\e^{-itH_0}]\zeta(H_0)
+[D_\Sigma,\zeta(H)]\e^{itH}\eta(H)\e^{-itH_0}\zeta(H_0)\\
&\quad+\zeta(H)\e^{itH}\eta(H)\e^{-itH_0}[D_\Sigma,\zeta(H_0)],
\end{align*}
and the last two operators belong to $\B(H)$ with norm uniformly bounded in $t$. Let
$\varphi,\psi\in\dom(D_\Sigma)$. Using Lemma \ref{D_Sigma}.(b) and Lemma \ref{R_lambda}.(b) one
gets for the first operator the following equalities
\begin{align*}
&\<\varphi,\zeta(H)[D_\Sigma,\e^{itH}\eta(H)\e^{-itH_0}]\zeta(H_0)\psi\>\\
&=\<\varphi,\zeta(H)[D_\Sigma,\e^{itH}]\eta(H)\e^{-itH_0}\zeta(H_0)\psi\>\\
&\quad+\<\varphi,\zeta(H)\e^{itH}[D_\Sigma,\eta(H)]\e^{-itH_0}\zeta(H_0)\psi\>\\
&\quad+\<\varphi,\zeta(H)\e^{itH}\eta(H)[D_\Sigma,\e^{-itH_0}]\zeta(H_0)\psi\>\\
&=-\int_0^t\d s\,\big\langle\varphi,\zeta(H)\e^{i(t-s)H}i[H,D_\Sigma]^a
\e^{isH}\eta(H)\e^{-itH_0}\zeta(H_0)\big\rangle\\
&\quad+\<\varphi,\zeta(H)\e^{itH}[D_\Sigma,\eta(H)]\e^{-itH_0}\zeta(H_0)\psi\>\\
&\quad+t\<\varphi,\zeta(H)\e^{itH}\eta(H)\e^{-itH_0}f(H_0)\zeta(H_0)\psi\>\\
&=\int_0^t\d s\,\big\langle\varphi,\zeta(H)\e^{i(t-s)H}\V_{\Sigma,f}
\e^{isH}\eta(H)\e^{-itH_0}\zeta(H_0)\big\rangle\\
&\quad+\<\varphi,\zeta(H)\e^{itH}[D_\Sigma,\eta(H)]\e^{-itH_0}\zeta(H_0)\psi\>\\
&\quad-t\<\varphi,\eta(H)\e^{itH}\{f(H)-f(H_0)\}\e^{-itH_0}\zeta(H_0)\psi\>.
\end{align*}
The first two terms are bounded by $\textsc c\!\;\|\varphi\|\cdot\|\psi\|$ with $\textsc c>0$
independent of $\varphi,\psi$ and $t$ (use the local $H$-smoothness of $\V_{\Sigma,f}$ for the
first term). Furthermore due to the local $H$- and $H_0$-smoothness of $f(H)-f(H_0)$ one can find
a sequence $t_n\to\infty$ as $n\to\infty$ such that
$$
\lim_{n\to\infty}
t_n\<\varphi,\eta(H)\e^{it_nH}\{f(H)-f(H_0)\}\e^{-it_nH_0}\zeta(H_0)\psi\>=0.
$$
This together with the previous remarks implies that
$$
\lim_{n\to\infty}\<\varphi,[D_\Sigma,\zeta(H)\e^{it_nH}\eta(H)\e^{-it_nH_0}\zeta(H_0)]\psi\>
\le\textsc c'\|\varphi\|\cdot\|\psi\|,
$$
with $\textsc c'>0$ independent of $\varphi,\psi$ and $t$. Thus using the intertwining relation
and the identity $\eta(H_0)=\zeta(H_0)\eta(H_0)\zeta(H_0)$ one finds that
\begin{align*}
&\big|\<D_\Sigma\varphi,W_+\eta(H_0)\psi\>-\<\varphi,W_+\eta(H_0)\psi\>\big|\\
&=\lim_{n\to\infty}
\big|\<\varphi,[D_\Sigma,\zeta(H)\e^{it_nH}\eta(H)\e^{-it_nH_0}\zeta(H_0)]\psi\>\big|\\
&\le\textsc c'\|\varphi\|\cdot\|\psi\|.
\end{align*}
This proves the result for $[D_\Sigma,W_+\eta(H_0)]$. A similar proof holds for
$[D_\Sigma,W_-\eta(H_0)]$. Since the wave operators are complete, one has $W_\pm^*\eta(H)=\textrm{s-}\lim_{t\to\pm\infty}\e^{itH_0}\e^{-itH}\eta(H)$, and an analogous
proof can be given for the operators $[D_\Sigma,W_\pm^*\eta(H)]$.
\end{proof}

\begin{Remark}\label{bibo}
In the case $\Sigma=\mathcal B$ the requirements of Definition \ref{F_Sigma} and Assumption
\ref{difference} are satisfied by many functions $f$. A natural choice is $f(u)=2u$, $u\in\R$,
since in such a case $f(H)-f(H_0)=2V\in\B\big(\H^2_t,\H_{t+\kappa}\big)$, $t\in\R$, $\kappa>1$. If
$\Sigma$ is not spherical there are still many appropriate choices for $f$. For instance if
$\gamma>0$, then the function $f(u)=2(u^2+\gamma)^{-1}u^3$, $u\in\R$, satisfies all the desired
requirements. Indeed in such a case one has on $\H^2$ the following equalities
\begin{align*}
&f(H)-f(H_0)\\
&=2V-2\gamma\big[(H^2+\gamma)^{-1}H-(H_0^2+\gamma)^{-1}H_0\big]\\
&=2V-2\gamma(H^2+\gamma)^{-1}V+2\gamma(H^2+\gamma)^{-1}(H_0V+VH_0+V^2)(H_0^2+\gamma)^{-1}H_0,
\end{align*}
and thus $f(H)-f(H_0)$ also extends to an element of $\B\big(\H^2_t,\H_{t+\kappa}\big)$, $t\in\R$,
$\kappa>1$, due to Lemma \ref{Hminus} and the assumptions on $V$.
\end{Remark}

Next Theorem provides a rigorous meaning to the anisotropic Lavine's formula \eqref{pizza_1}.

\begin{Theorem}\label{lavine}
Let $\Sigma$ satisfy Assumption \ref{Sigma}. Let $V$ satisfy Assumption \ref{potential} with
$\kappa>4$. Suppose that Assumption \ref{difference} is verified. Then one has for each
$\varphi\in\D_s$ with $s>2$
\begin{equation}\label{new_Lavine}
\lim_{r\to\infty}\tau_r(\varphi)
=\int_{-\infty}^\infty\d s\,\big\langle \e^{-isH}W_-f(H_0)^{-1/2}\varphi,
\V_{\Sigma,f}\e^{-isH}W_-f(H_0)^{-1/2}\varphi\big\rangle_{2,-2},
\end{equation}
where $\<\;\!\cdot\;\!,\;\!\cdot\;\!\>_{\scriptscriptstyle2,-2}:\H^2\times\H^{-2}\to\C$ is the
anti-duality map between $\H^2$ and $\H^{-2}$.
\end{Theorem}

\begin{proof}
(i) Set $W(t):=\e^{itH}\e^{-itH_0}$, and let $\psi:=\eta(H)\widetilde\psi$, where
$\eta\in C^\infty_0((0,\infty)\setminus\sigma_{\rm pp}(H))$ and $\widetilde\psi\in\dom(D_\Sigma)$.
We shall prove that $\|D_\Sigma W(t)^*\psi\|\le\textsc c$, with $\textsc c$ independent of $t$.
Due to Lemma \ref{D_Sigma}.(b) and Lemma \ref{D_SigmaV}.(c) one has
\begin{align}
\|D_\Sigma W(t)^*\psi\|
&=\big\|\e^{-itH_0}D_\Sigma\e^{itH_0}\e^{-itH}\eta(H)(D_\Sigma+i)^{-1}\psi_1\big\|\nonumber\\
&\le|t|\big\|\{f(H)-f(H_0)\}\e^{-itH}\eta(H)(D_\Sigma+i)^{-1}\psi_1\big\|\label{mydoudou}\\
&\quad+\big\|\{D_\Sigma-tf(H)\}\e^{-itH}\eta(H)(D_\Sigma+i)^{-1}\psi_1\big\|\nonumber,
\end{align}
where $\psi\equiv\eta(H)(D_\Sigma+i)^{-1}\psi_1$. Let $z\in\C\setminus\{\sigma(H_0)\cup\sigma(H)\}$
and set $\widetilde\eta(H):=(H-z)^2\eta(H)$. Then Lemmas \ref{D_Sigma}.(a), \ref{Hminus}, and
\ref{giantcoucou}.(a) imply that
\begin{align*}
&|t|\big\|\{f(H)-f(H_0)\}\e^{-itH}\eta(H)(D_\Sigma+i)^{-1}\psi_1\big\|\\
&\le|t|\big\|\{f(H)-f(H_0)\}(H-z)^{-2}(D_\Sigma+i)\big\|\cdot
\big\|(D_\Sigma+i)^{-1}\e^{-itH}\widetilde\eta(H)(D_\Sigma+i)^{-1}\big\|\\
&\le{\rm Const.}
\end{align*}
Calculations similar to those of Lemma \ref{giantcoucou}.(a) show that the second term of
\eqref{mydoudou} is also bounded uniformly in $t$.

(ii) Let $W(t)$ and $\psi$ be as in point (i). Lemma \ref{D_Sigma}.(b), Lemma \ref{D_SigmaV}.(c),
and commutator calculations as in \eqref{erg} lead to
\begin{align*}
\<W(t)^*\psi,D_\Sigma W(t)^*\psi\>
&=\<\psi,\e^{itH}D_\Sigma\e^{-itH}\psi\>-t\<\psi,\e^{itH}f(H_0)\e^{-itH}\psi\>\\
&=\<\psi,D_\Sigma\psi\>
-\int_0^t\d s\,\big\langle\e^{-isH}\psi,\V_{\Sigma,f}\e^{-isH}\psi\big\rangle_{2,-2}\\
&\quad+t\<\psi,\e^{itH}\{f(H)-f(H_0)\}\e^{-itH}\psi\>.\\
\end{align*}
The local $H$-smoothness of $f(H)-f(H_0)$ implies the existence of a sequence $t_n\to\infty$ as
$n\to\infty$ such that
$$
\lim_{n\to\infty}t_n\<\psi,\e^{it_nH}\{f(H)-f(H_0)\}\e^{-it_nH}\psi\>=0.
$$
This together with point (i) and the local $H$-smoothness of $\V_{\Sigma,f}$ implies that
$$
\<W^*_+\psi,D_\Sigma W^*_+\psi\>=\<\psi,D_\Sigma\psi\>
-\int_0^\infty\d s\,\big\langle\e^{-isH}\psi,\V_{\Sigma,f}\e^{-isH}\psi\big\rangle_{2,-2}.
$$
Similarly, one finds
$$
\<W^*_-\psi,D_\Sigma W^*_-\psi\>
=\<\psi,D_\Sigma\psi\>
+\int_{-\infty}^0\d s\,\big\langle\e^{-isH}\psi,\V_{\Sigma,f}\e^{-isH}\psi\big\rangle_{2,-2},
$$
and thus
\begin{equation}\label{globibulga}
\<W^*_+\psi,D_\Sigma W^*_+\psi\>-\<W^*_-\psi,D_\Sigma W^*_-\psi\>=-\int_{-\infty}^\infty\d s\,
\big\langle\e^{-isH}\psi,\V_{\Sigma,f}\e^{-isH}\psi\big\rangle_{2,-2}.
\end{equation}
Let $\varphi\in\D_s$ with $s>2$. Due to Lemma \ref{giantcoucou}.(b) the vector
$W_-f(H_0)^{-1/2}\varphi$ is of the form $\eta(H)\widetilde\psi$, with
$\eta\in C^\infty_0((0,\infty)\setminus\sigma_{\rm pp}(H))$ and $\widetilde\psi\in\dom(D_\Sigma)$.
Thus one can put $\psi=W_-f(H_0)^{-1/2}\varphi$ in Formula \eqref{globibulga}. This gives
\begin{align*}
&\big\langle Sf(H_0)^{-1/2}\varphi,D_\Sigma Sf(H_0)^{-1/2}\varphi\big\rangle
-\big\langle f(H_0)^{-1/2}\varphi,D_\Sigma f(H_0)^{-1/2}\varphi\big\rangle\\
&=-\int_{-\infty}^\infty\d s\,\big\langle\e^{-isH}W_-f(H_0)^{-1/2}\varphi,
\V_{\Sigma,f}\e^{-isH}W_-f(H_0)^{-1/2}\varphi\big\rangle_{2,-2},
\end{align*}
and the claim follows by Theorem \ref{time_delay}.
\end{proof}

\begin{Remark}
Symmetrised time delay and usual time delay are equal when $\Sigma$ is spherical (see Formula
\eqref{sweetie}). Therefore in such a case Formula \eqref{new_Lavine} must reduces to the usual
Lavine's formula. This turns out to be true. Indeed if $\Sigma=\mathcal B$ and $f(u)=2u$, then
$f(H_0)=2H_0$, $\V_{\Sigma,f}$ is equal to the virial $\widetilde V:=2V-i[V,D]^a$, and Formula
\eqref{new_Lavine} takes the usual form
$$
\lim_{r\to\infty}\tau_r(\varphi)
=\int_{-\infty}^\infty\d s\,\big\langle \e^{-isH}W_-H_0^{-1/2}\varphi,
\big\{V-\textstyle\frac i2[V,D]^a\big\}\e^{-isH}W_-H_0^{-1/2}\varphi\big\rangle_{2,-2}.
$$
\end{Remark}

In the following remark we give some insight on the meaning of Formula \eqref{new_Lavine} when
$\Sigma$ is not spherical. Then we present two simple examples as an illustration.

\begin{Remark}\label{non_isotropic}
Let $V$ satisfy Assumption \ref{potential} with $\kappa>4$, and choose a set
$\Sigma\ne\mathcal B$ satisfying Assumption \ref{Sigma}. In such a case the function
$f_\gamma(u):=2(u^2+\gamma)^{-1}u^3$, $u\in\R$, fulfills the requirements of Definition
\ref{F_Sigma} and Assumption \ref{difference} (see Remark \ref{bibo}). Thus Theorem \ref{lavine}
applies, and one has for $\varphi\in\D_s$ with $s>2$
\begin{align*}
&\lim_{r\to\infty}\tau_r(\varphi)\\
&=\lim_{\gamma\searrow0}\int_{-\infty}^\infty\d s
\,\big\langle\e^{-isH}W_-f_\gamma(H_0)^{-1/2}\varphi,
\V_{\Sigma,f_\gamma}\e^{-isH}W_-f_\gamma(H_0)^{-1/2}\varphi\big\rangle_{2,-2}.
\end{align*}
Now $f_\gamma(H_0)\varphi$ converges in norm to $2H_0\varphi$ as $\gamma\searrow0$, so formally
one gets the identity
\begin{equation}\label{sign}
\lim_{r\to\infty}\tau_r(\varphi)
=\12\int_{-\infty}^\infty\d s\,\big\langle \e^{-isH}W_-H_0^{-1/2}\varphi,
\V_\Sigma\e^{-isH}W_-H_0^{-1/2}\varphi\big\rangle_{2,-2},
\end{equation}
where
$$
\V_\Sigma:=2V-i[V,D_\Sigma]^a
=2V-{\textstyle\frac i2}\sum_{j\le d}\big\{\big[V,{F_\Sigma}_j(P)\big]\cdot Q_j
+Q_j\cdot\big[V,{F_\Sigma}_j(P)\big]\big\},
$$
and
\begin{equation}\label{F_Sigma_P2}
{F_\Sigma}_j(P)=-(\partial_jG_\Sigma)(P)P^2.
\end{equation}

The pseudodifferential operator $\V_\Sigma$ generalises the virial $\widetilde V$ of the
isotropic case. It furnish a measure of the variation of the potential $V$ along the vector field
$-F_\Sigma$, which is orthogonal to the hypersurfaces $\partial\Sigma_r$ due to Remark
\ref{orthogonal}. Therefore Formula \eqref{sign} establishes a relation between symmetrised time
delay and the variation of $V$ along $-F_\Sigma$. Moreover one can rewrite $\V_\Sigma$ as
\begin{align*}
\V_\Sigma
&=\widetilde V+i[V,D-D_\Sigma]^a\\
&=\widetilde V+{\textstyle\frac i2}\sum_{j\le d}
\big\{\big[V,\big(P_j-{F_\Sigma}_j(P)\big)\big]\cdot Q_j
+Q_j\cdot\big[V,\big(P_j-{F_\Sigma}_j(P)\big)\big]\big\}.
\end{align*}
where $P-F_\Sigma(P)$ is orthogonal to $P$ due to Formulas \eqref{F_Sigma_P2} and \eqref{moinsun}.
Consequently there are two distinct contributions to symmetrised time delay. The first one is
standard; it is associated to the term $\widetilde V$, and it is due to the variation of the
potential $V$ along the radial coordinate (see \cite[Sec. 6]{Lav74} for details). The second one
is new; it is associated to the term $i[V,D-D_\Sigma]^a$ and it is due to the variation of $V$
along the vector field $x\mapsto x-F_\Sigma(x)$.
\end{Remark}

\begin{Example}[Examples in $\R^2$]
Set $d=2$, suppose that $V$ satisfies Assumption \ref{potential} with $\kappa>4$, and let $\Sigma$
be equal to the superellipse $\E:=\big\{(x_1,x_2)\in\R^2\mid x_1^4+x_2^4<1\big\}$. Then one has $G_\E(x)=-\frac14\ln\big(x_1^4+x_2^4\big)$ and
$(\partial_jG_\E)(x)=-x_j^3\big(x_1^4+x_2^4\big)^{-1}$. Thus due to Remark \ref{non_isotropic} the
symmetrised time delay associated to $\E$ is (formally) caracterised by the pseudodifferential
operator
$$
\V_\E
=2V-{\textstyle\frac i2}\sum_{j\le d}\big\{\big[V,{F_\E}_j(P)\big]\cdot Q_j
+Q_j\cdot\big[V,{F_\E}_j(P)\big]\big\},
$$
where ${F_\E}_j(P)=P_j^3P^2\big(P_1^4+P_2^4\big)^{-1}$ (see Figure \ref{F_carre}).

\begin{figure}[htbp]
\begin{center}
\includegraphics[angle=90,scale=0.3]{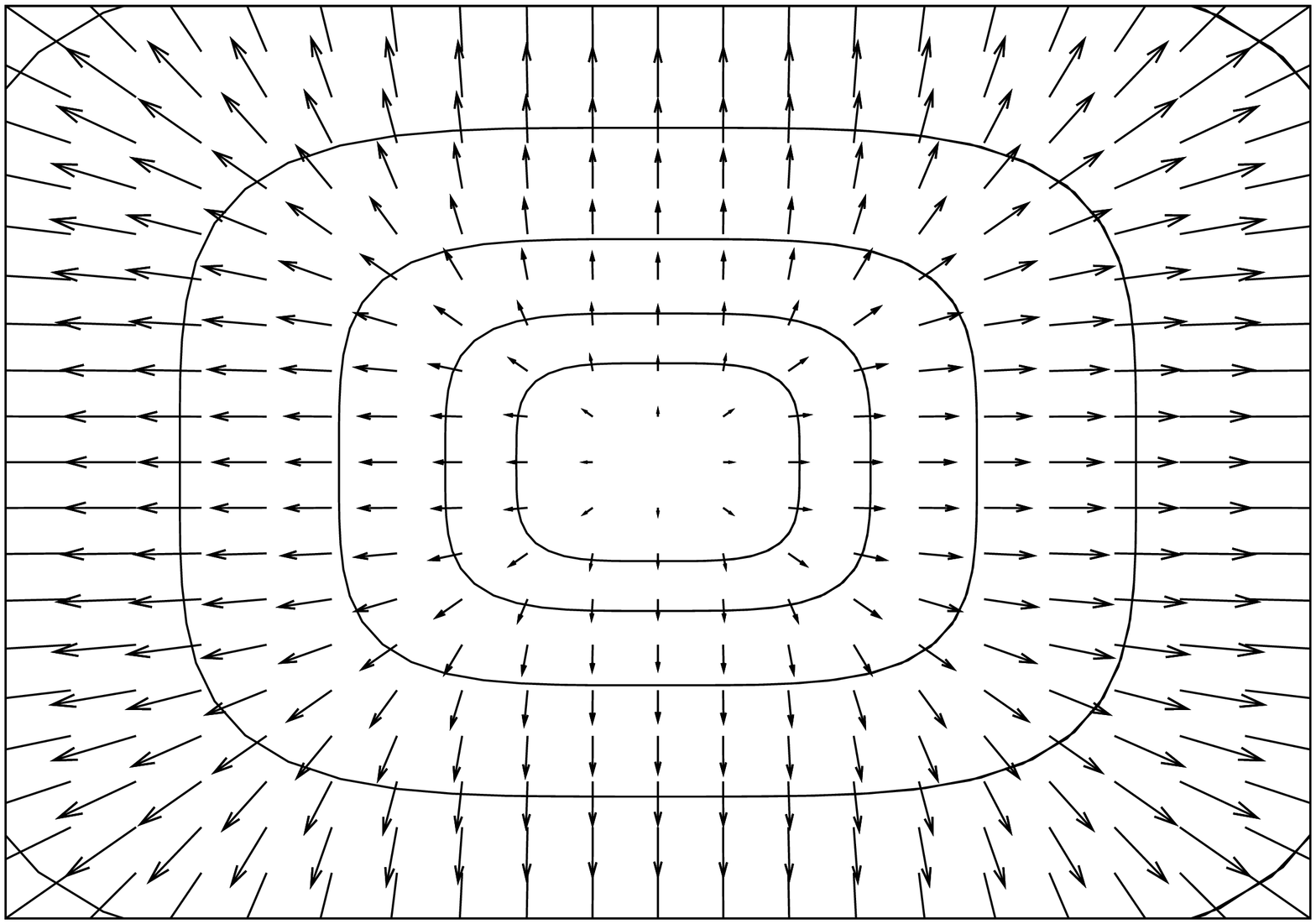}
\caption{\textsf{\footnotesize The vector field $F_\E$ and the sets $\partial\E_r$}}
\label{F_carre}
\vspace{-10pt}
\end{center}
\end{figure}

When $\Sigma$ is equal to the star-type set
$$
\mathcal S:=\Big\{\ell(\theta)\e^{i\theta}\in\R^2\mid
\theta\in[0,2\pi),~\ell(\theta)<\big[\cos(2\theta)^8+\sin(2\theta)^8\big]^{-1/2}\Big\},
$$
one has
$G_{\mathcal S}(x)=\frac72\ln(x_1^2+x_2^2)-\frac12\ln\big[(x_1^2-x_2^2)^8+2^8(x_1x_2)^8\big]$,
and a direct calculation using Formula \eqref{F_Sigma_P2} gives the vector field $F_{\mathcal S}$.
The result is plotted in Figure \ref{F_etoile}.

\begin{figure}[htbp]
\begin{center}
\includegraphics[angle=90,scale=0.3]{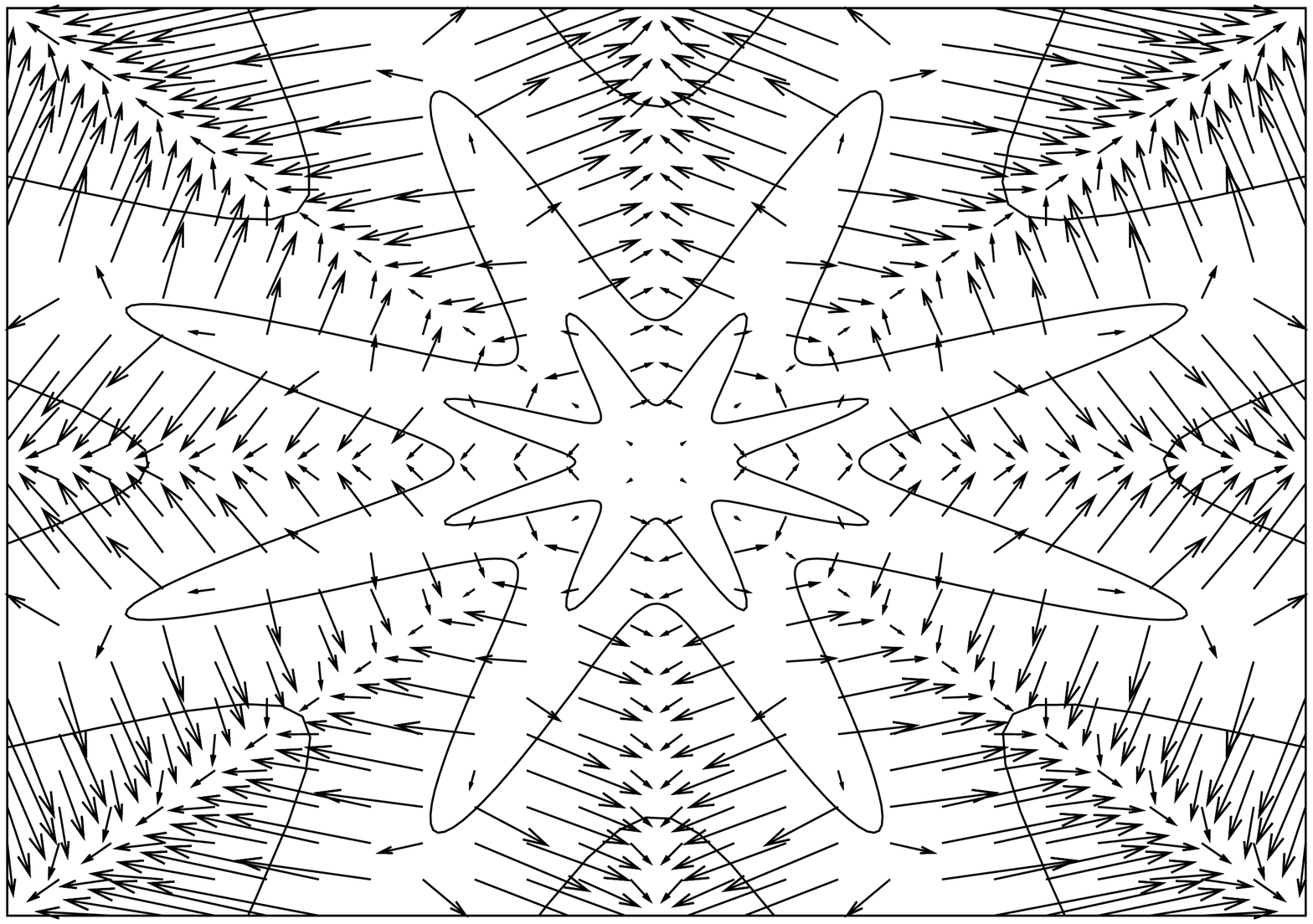}
\caption{\textsf{\footnotesize The vector field $F_{\mathcal S}$ and the sets
$\partial\mathcal S_r$}}
\label{F_etoile}
\vspace{-10pt}
\end{center}
\end{figure}

\end{Example}

\section*{Acknowledgements}

The author thanks the Swiss National Science Foundation and the Department of Mathematics of the
University of Cergy-Pontoise for financial support.

\section*{Appendix}

\begin{proof}[Proof of Lemma \ref{Hminus}]
We first prove that $(H-z)^{-1}$ extends to an element of $\B\big(\H^{-2}_t,\H_t\big)$ for each
$t\ge0$. This clearly holds for $t=0$. Since $(H_0-z)^{-1}\<P\>^2=2+(1+2z)(H_0-z)^{-1}$ one has by
virtue of the second resolvent equation
\begin{align}
&\<Q\>^t(H-z)^{-1}\<P\>^2\<Q\>^{-t}\label{t_iteration}\\
&=2+(1+2z)\<Q\>^t(H_0-z)^{-1}\<Q\>^{-t}\nonumber\\
&\quad-\<Q\>^t(H_0-z)^{-1}(\<Q\>V)\<Q\>^{-t}\cdot\<Q\>^{t-1}(H-z)^{-1}\<P\>^2\<Q\>^{-t}.
\nonumber
\end{align}
If we take $t=1$ we find that each term on the r.h.s. of \eqref{t_iteration} is in $\B(\H)$ due to
\cite[Lemmas 1 \& 2]{ACS}. Hence, by interpolation, $\<Q\>^t(H-z)^{-1}\<P\>^2\<Q\>^{-t}\in\B(\H)$
for each $t\in[0,1]$. Next we choose $t\in(1,2]$ and obtain, by using the preceding result and
\eqref{t_iteration}, that $\<Q\>^t(H-z)^{-1}\<P\>^2\<Q\>^{-t}\in\B(\H)$ for these values of $t$.
By iteration (take $t\in(2,3]$, then $t\in(3,4]$, etc.) one obtains that
$\<Q\>^t(H-z)^{-1}\<P\>^2\<Q\>^{-t}\in\B(\H)$ for each $t>0$. Thus $(H-z)^{-1}$ extends to an
element of $\B\big(\H^{-2}_t,\H_t\big)$ for each $t\ge0$. A similar argument shows that
$(H-z)^{-1}$ also extends to an element of $\B\big(\H^{-2}_t,\H_t\big)$ for each $t<0$. The claim
follows then by using duality and interpolation.
\end{proof}

\begin{proof}[Proof of Lemma \ref{cond_lone}]
For $\varphi\in\D_s$ and $t\in\R$, we have (see the proof of \cite[Lemma 4.6]{Jensen81})
$$
\(W_--1\)\e^{-itH_0}\varphi
=-i\e^{-itH}\int_{-\infty}^t\d\tau\,\e^{i\tau H}V\e^{-i\tau H_0}\varphi,
$$
where the integral is strongly convergent. Hence to prove \eqref{R-} it is enough to show that
\begin{equation}\label{l1 condition}
\int_{-\infty}^{-\delta}\d t\int_{-\infty}^t\d\tau\left\|V\e^{-i\tau H_0}\varphi\right\|<\infty
\end{equation}
for some $\delta>0$. If $\zeta:=\min\{\kappa,s\}$, then $\big\|\<Q\>^\zeta\varphi\big\|<\infty$,
and $V\<P\>^{-2}\<Q\>^\zeta$ belongs to $\B(\H)$ due to Assumption \ref{potential}. Since
$\eta(H_0)\varphi=\varphi$ for some $\eta\in C^\infty_0((0,\infty)\setminus\sigma_{\rm pp}(H))$,
this implies that
$$
\left\|V\e^{-i\tau H_0}\varphi\right\|
\leq\textrm{Const.}\,\big\|\<Q\>^{-\zeta}\<P\>^2\eta(H_0)\e^{-i\tau H_0}\<Q\>^{-\zeta}\big\|.
$$
For each $\varepsilon>0$, it follows from \cite[Lemma 9]{ACS} that there exists a constant
$\textsc c>0$ such that
$\left\|V\e^{-i\tau H_0}\varphi\right\|\leq\textsc c\(1+|\tau|\)^{-\zeta+\varepsilon}$. Since
$\zeta>2$, this implies \eqref{R-}. The proof of \eqref{R+} is similar.
\end{proof}


\end{document}